\documentclass[preprint,12pt]{elsarticle}

\usepackage{hyperref}
\usepackage{amstext}
\usepackage{graphicx}
\usepackage{epstopdf}
\usepackage{amssymb}
\usepackage{mathrsfs}
\usepackage{amsmath}
\usepackage{amsthm}
\usepackage{amsfonts}
\usepackage{verbatim}
\usepackage{color}
\usepackage{natbib}

\makeatletter
\newcommand{\myvec}[2][r]{%
  \gdef\@VORNE{1}
  \left(\hskip-\arraycolsep%
    \begin{array}{#1}\vekSp@lten{#2}\end{array}%
  \hskip-\arraycolsep\right)}

\def\vekSp@lten#1{\xvekSp@lten#1;vekL@stLine;}
\def\vekL@stLine{vekL@stLine}
\def\xvekSp@lten#1;{\def\temp{#1}%
  \ifx\temp\vekL@stLine
  \else
    \ifnum\@VORNE=1\gdef\@VORNE{0}
    \else\@arraycr\fi%
    #1%
    \expandafter\xvekSp@lten
  \fi}
\makeatother

\newtheorem{theorem}{Theorem}[section]

\def\curl{\mathrm{curl}~}
\DeclareMathOperator{\sech}{sech}
\def\dd{\displaystyle}

\def\8{\infty}

\def\p{\partial}

\journal{Journal of Computational Physics}

\begin{document}

\begin{frontmatter}

\title{Geodesic curvature driven surface microdomain formation}


\author{Melissa R. Adkins, and Y. C. Zhou\fnref{myfootnote}}
\address{Department of Mathematics, Colorado State University, Fort Collins, CO 80523-1874}
\fntext[myfootnote]{Corresponding author. Email: yzhou@math.colostate.edu}



\begin{abstract}
Lipid bilayer membranes are not uniform and clusters of lipids in a more ordered state exist within the generally disorder 
lipid milieu of the membrane. These clusters of ordered lipids microdomains are now referred to as lipid rafts. Recent
reports attribute the formation of these microdomains to the geometrical and molecular mechanical mismatch of lipids 
of different species on the boundary. Here we introduce the geodesic curvature to characterize the geometry of the
domain boundary, and develop a geodesic curvature energy model to describe the formation of these microdomains
as a result of energy minimization. Our model accepts the intrinsic geodesic curvature of any binary
lipid mixture as an input, and will produce microdomains of the given geodesic curvature as demonstrated by
three sets of numerical simulations. Our results are in contrast to the surface phase separation predicted
by the classical surface Cahn-Hilliard equation, which tends to generate large domains as a result of the
minimizing line tension. Our model provides a direct and quantified description of the structure 
inhomogeneity of lipid bilayer membrane, and can be coupled to the investigations of biological 
processes on membranes for which such inhomogeneity plays essential roles. 
\end{abstract}

\begin{keyword}
lipid bilayer membrane; lipid domain; geodesic curvature; surface phase separation; phase field method; variational principle;
numerical simulations 
\MSC[2014] 53B10 \sep 65D18 \sep 92C15
\end{keyword}

\end{frontmatter}


\section{Introduction}

This work is motivated by the formation of lipid rafts in lipid membranes. Lipid bilayer membranes are of utmost importance for the survival of cells. 
They separate the interior of cells from the extracellular environment
and compartmentalize subcellular organelles so suitable micro-environment can be maintained in the enclosed domains for various vital biochemical 
and biophysical reactions. They are material basis for morphological changes such as budding, tubulation, fission and fussion that occur during cell 
division, biological reproduction, and intracellular membrane trafficking. They also provide a physical platform to store and transduce energy as
electrochemical gradients, to segregate or disperse particular membrane proteins, and to act as messengers in signal transduction and molecular 
recognition processes \cite{MeerG2008a}. While most of these functionalities depends on the fluidity of the lipids and thereby the free diffusion 
of lipids and proteins in the bilayer, accumulated evidences show that lipids and proteins on bilayer membranes segregates into discrete domains
of distinct composition and various sizes \cite{KarnovskyM1982a,SimonK2000a,VeatchS2003a,KusumiA2005a,VieiraF2010a,WangT2012a}. 
The domain boundaries can appear as 
the barriers of free lateral diffusion of lipids and proteins, as the measured lateral diffusion coefficients of lipids and proteins {\it in vivo}
are less than the measured coefficients in artificial pure bilayer by more than one order of magnitude \cite{KusumiA2005a,TrimbleW2015a}. Inside the 
domains and on the domain boundaries particular proteins may aggregate, to cause various membrane curvature as the consequence of the modification of local
membrane composition \cite{HuttnerW2001a,RennerL2011a,RyuY2015a} or to complete specific signal transduction \cite{HinderliterA2001a,GorfeA2010a,YeungT2008a,ZhdanovV2015a}.
Some of these domains are transient, with a duration ranging from seconds to minutes, some can persist for the entire life of the cell, 
and the domains themselves can diffuse on the membrane surface as well \cite{CicutaP2007a}. The composition, location, size, configuration, duration 
of these domains and the dynamics of these characteristics are of functional and structural significance to the associated biological processes. Efforts 
integrating direct microscopic measurements, biophysical modeling, and computational simulations have been invested to elucidate the underlying physics 
of the dynamics of these lipid domains and predict their biological consequences \cite{PikeL2003a,NicolauD2006a,HancockJ2006a}. Before introducing our 
approach based on the geodesic curvature energy of the lipid domain boundaries we first review four most representative theoretical studies on the 
dynamics of lipid domains.

Lipid domains may appear as a result of lipid phase separation caused by distinct spontaneous curvatures. When bilayer membranes have multiple 
lipid species of distinct spontaneous curvatures, individual lipid species may be localized to regions where the local mean curvatures best approximate 
the corresponding spontaneous curvatures of the 
residing lipid species \cite{BaumgartT2003a}. Wang and Du formalized this reasoning by summing up the classical Canham-Evans-Helfrich 
energy \cite{CanhamP1970,EvansE1974,HelfrichW1973} for each individual lipid species and the line tension energy to generate a multi-component lipid 
membrane model \cite{WangX2008a}. By representing the membrane bending energy using the phase field formulation, they have
obtained rich patterns of membrane morphology and the generation of lipid membrane domains of different mean curvatures, where lipid 
species of the approximate spontaneous curvatures are concentrated. This model was also extended to simulate the open membrane thanks to 
the line tension energy, and the closing of membrane pores was simulated corresponding to the vanishing linear tension energy. These 
permanent domains have sizes that are determined by 
local mean curvatures of the membrane necks or bumps. These sizes in general do not match the measured sizes of mobile 
lipid rafts \cite{YuanC2002a,PikeL2003a}.

The classical phase separation model based on the Ginzburg-Landau (GL) free energy could also be directly applied on a 
membrane surface to generate surface phase separation, and the results can be related to the lipid domains. 
A surface Cahn-Hilliard equation can be derived for the gradient flow of the GL free energy, and the numerical simulations will 
produce large separated domains as a result of the coarsening dynamics \cite{DuQ2011a}. In order to generate small domains at 
spatial and temporal scales comparable
to experimental results, Camley and Brown couples the GL free energy for quasi two-dimensional binary lipids 
mixtures to the random hydrodynamics and thermal fluctuations \cite{CamleyB2010a,CamleyB2011a}. The random in-plane velocity field of the 
membrane is given by Saffman-Delbruck hydrodynamic model \cite{SaffmanP1975a}. This velocity field is added to the 
Cahn-Hilliard equation for the gradient flow of the GL free energy to produce an advection-diffusion equation, above which 
a Gaussian white noise is added, modeling the thermal fluctuation as a random source to the order parameter. Complete
phase separation shall occur as the result of a sequence of coarsening dynamics when the GL free energy is minimized, while 
the domain boundaries flicker as a result of random hydrodynamic and thermal perturbations, with a flickering
magnitude depending on the competing between the random perturbation and the persisting linear tension. 
Under high line tension small domains will merge to form large separated domain, but small domains under a critical size could remain 
separated for a long time during the course of coarsening if the line tension is not large enough to suppress the random perturbation, 
giving rise to lasting microdomains. Various dynamical scaling rates were summarized to related the microdomain size and the time when 
the domain size is far way from the Saffmann-Delbruck length $L_{sd}$ determined by the relative viscosity of the lipid membrane with 
respect to the surrounding fluid field. This approach has been recently extended to model multicomponent membranes with 
embedded proteins \cite{NoruzifarE2014a,CamleyB2014a}.

It is also possible to simulate lipid microdomains in the deterministic setting. Arguing that lipid rafts are microdomains of lipids compactly 
organized around embedded protein receptors, Witkowski, Backofen and Voigt proposed to supplement the classical Ginzburg-Landau free energy with Gaussian
potentials localized at specified positions where the membrane proteins are supposed to be embedded \cite{WitkowskiT2012a}. By specifying the 
center and modulation of these external potentials they were able to produce lipid microdomains of arbitrary size at arbitrary position. 
Coarsening dynamics were reproduced by solving the Cahn-Hilliard equation for the gradient flow of the total energy, and a scaling law 
was deduced for the growth of the microdomains. In contrast to the above approaches, one is not able to drive the lateral diffusion coefficient
of the microdomains as their positions are specified in the construction of the free energy. In general this model lacks 
a biophysical interpretation of the external potential and the related parameters that could justify the striking generation of 
lipid microdomains in the absence of line tension. 

In addition to the above continuum approaches, particle-based discrete methods have also been developed to simulate the lipid microdomains.
Molecular dynamics simulations, fully atomic or coarse grained, as reviewed in \cite{BagatolliL2009a,BennettW2013a}, 
have been able to generate lipid membrane domains that could interpret some experiments on complex model membranes. Two discrete methods that gives 
particular valuable insight into the structure and dynamics of lipid microdomains are dynamical triangulation Monte Carlo (DTMC) \cite{KumarP1996a} 
and dispersive particle dynamics (DPD) \cite{SchilcockJ2002a,LaradjiM2004a}. DTMC neglects the solvent hydrodynamics and approximate a bilayer 
membrane as a randomly triangulated sheet. Each vertex is described by a three-dimensional position vector, and all vertices are connected by 
flexible tethers, which flip during the course of dynamical triangulation to simulate phase segregation. DPD adopts a coarse-grained representation 
of amphiphilic lipids as connected head (H) and tail (C$_n$) beads, with an variable number of tail beads. The geometric and molecular mechanics 
representations of lipids in DPD differ from the coarse-grained molecular dynamics simulations (those based on the MARTINI force 
field \cite{LiH2013a,RisseladaH2008a}, for example) in that all beads are soft with interaction defined by effective forces that reproduce 
the hydrodynamic behavior of fluid bilayer membrane rather than the classical intermolecular interactions. DPD allows
asymmetric lipid composition in the two leaflets, where different lateral sizes of the lipid domains can be simulated. While 
DPD can simulate the membrane properties to length and time scales that are unattainable by MD and coarse-grained MD simulations, 
there still exist gaps for these discrete methods to model large systems containing millions of amphiphiles over biologically 
relevant time scales, for the handling of which the continuum models have intrinsic advantages.

Here we work to develop a lipid micro-organization model by recognizing the intrinsic geodesic curvatures of the boundaries 
between lipids of different species (saturated and unsaturated, ordered and disordered lipids, too). Experimental observations and molecular 
dynamics simulations show that lipid microdomains have relatively stable size ranging from below 10nm to larger than 200 nm, 
depending on the lipid composition and temperature \cite{HancockJ2006a,RisseladaH2008a,BagatolliL2009a}. Modulo the perturbations to the 
boundaries of these microdomains caused by random forces due to the thermal fluctuations, hydrodynamic interaction etc, the 
size of an individual microdomain can be characterized by the curvature of boundary, which is the geodesic curvature since the 
microdomain is a patch on the three-dimensional surface. Consequently, one can identify an {\it intrinsic geodesic curvature} for a 
binary mixture of lipid species, similar to the identification of the spontaneous curvature of an individual lipid 
species. A geodesic curvature energy can be defined for a given binary mixture of lipid species on a membrane surface, in a 
way similar to the definition of Canham-Evans-Helfrich 
curvature energy. It is expected that the minimization of this geodesic curvature energy will generate the optimal interfaces of the
binary mixture of lipids. The minimization is conducted on the 2-manifold that represents the membrane surface, and these optimal 
interfaces define the boundaries of the lipid microdomains. Since the optimal interfaces are unknown {\it a priori}, and there
will be topological changes as the initial interfaces evolve during the course of curvature energy minimization, we adopt an Eulerian 
formulation of the geodesic curvature energy by relating the geodesic curvature to a phase field function on the membrane surface. We shall 
show that the phase field representation of the geodesic curvature energy can be regarded as a generalization of the phase field
modeling of the elastic bending curvature energy of bilayer membrane deformation \cite{DuQ2005d,DuQ2006a}. The equivalence between
the geodesic curvature energy and the phase field representation can be justified when the phase field function is given by
a properly scaled hyperbolic tangent function of the signed geodesic distance to the boundaries of lipid microdomains. The 
minimizer of the geodesic curvature energy in the Eulerian representation is obtained by evolving the gradient flow of the
phase field function. The forth-order evolutionary partial differential equation for the gradient flow are numerically solved 
by using a $C_0$ interior penalty discontinuous Galerkin method on the triangulated membrane surface. Computational simulations
using the geodesic curvature energy model on different surfaces demonstrate that our model can produce microdomains with specified
intrinsic geodesic curvatures.

The rest of the article is organized as follows. In Section \ref{sect:model} a brief review of the structure of lipid bilayer
membrane is followed by the introduction of geodesic curvature of the interfaces between saturated and unsaturated phases of
lipids. Geodesic curvature energy is then defined in Lagrangian and Eulerian formulations, respectively. Technical details of 
the justification of some critical properties of geodesic distance and geodesic curvature are given in Appendix \ref{sect:append}. 
The gradient flow of the geodesic curvature energy is derived using the energetical variational principle in Section \ref{sect:numeric}. 
We split the linear and nonlinear components of the governing equation for the gradient flow so that the linear components can 
be solved using an implicit time marching method and the convergence at each time step is obtained by efficient iterations over the 
nonlinear components. The proposed mathematical model and the related numerical methods are applied in Section \ref{sect:simulations} 
on a set of static surfaces with varying curvatures to examine their ability in simulating the generation of microdomains of 
specified geodesic curvatures. We summarize this work and outline its perspectives in the final Section \ref{sect:summary}.

\section{Mathematical Model of Lipid Mismatch and Geodesic Curvature Energy} \label{sect:model}

We consider a bilayer membrane with binary mixture of lipids. These can be two species of lipids, or lipids in saturated 
and unsaturated states, or lipids in ordered or disordered states. These two species of lipids have their own regular 
molecular geometry at the equilibrium state for a given temperature, modulo random thermal and hydrodynamic fluctuations. 
By regular we mean the uniform lateral spacing of lipid molecules, the ordered match of hydrophobic tails of two leaflets,
among other features. This regular molecular geometry dictates the intermolecular interactions among lipids of same species.
At the interface between domains of distinct lipid species, the regular molecular geometry of either species has to be
relaxed in a way such that the intermolecular interactions in the transitional region near the interface will fit
the different molecular geometry of the other species. This relaxation shall create curved interface between two species
in a manner similar to the creation of the surface tension at a fluid surface. Fig. \ref{fig:surface_pattern_model} illustrates the
mismatch of lipid molecular structure at the interface between saturated (s) and unsaturated (u) lipids. 
\begin{figure}[!ht]
\begin{center}
\includegraphics[height=2.8cm]{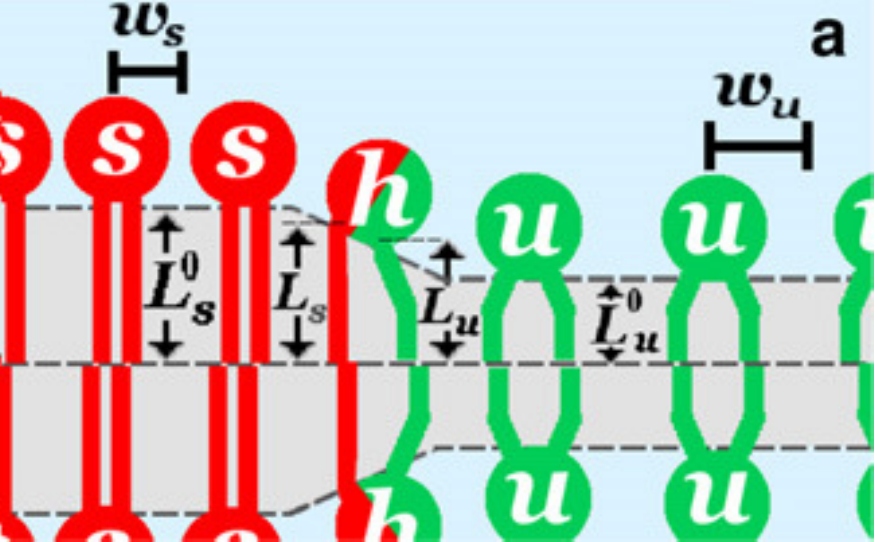} \hspace{3mm}
\includegraphics[height=2.8cm]{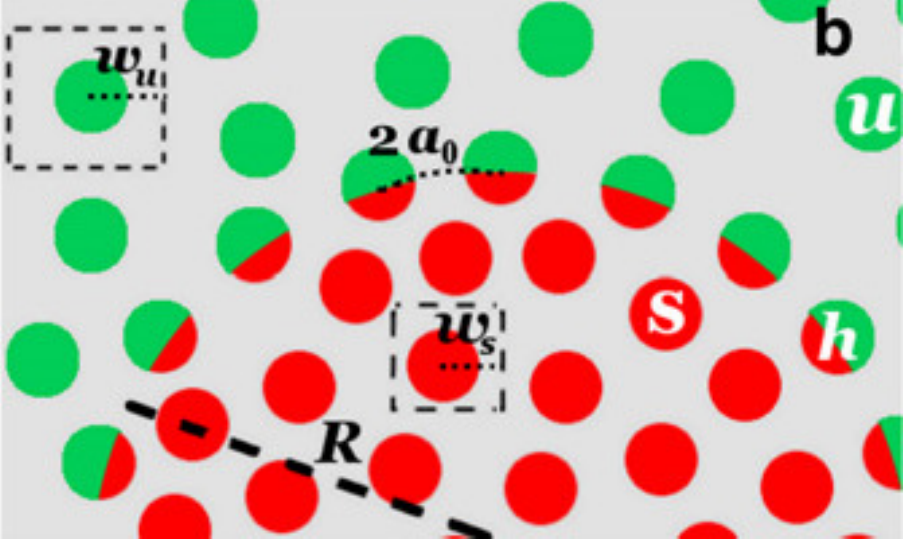} \hspace{3mm}
\includegraphics[height=2.8cm]{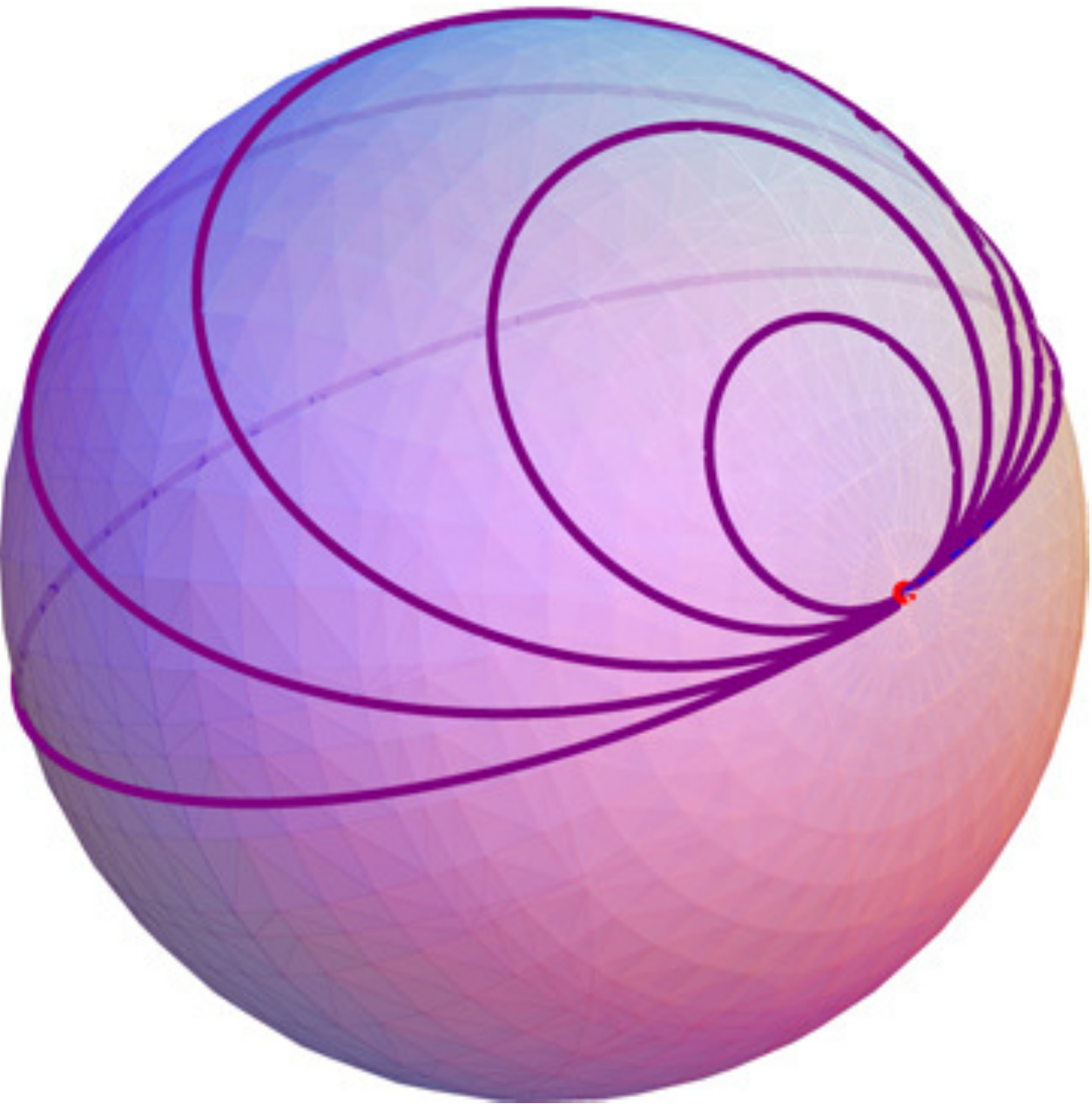}
\end{center}
\caption{Left: Mismatch of the lipid structures at the interface between two lipid domains \cite{BrewsterR2010a}. 
Used with permission. Middle: Within the transitional hybrid layer the otherwise regular lattices of the lipids 
in either domain relax to match each other, causing a bending interface \cite{BrewsterR2010a}. 
Used with permission. Right: Circles on a sphere have constant geodesic curvatures. The great circle has a vanishing 
geodesic curvature in particular.}
\label{fig:surface_pattern_model}
\end{figure}

\subsection{Geodesic curvature energy}
Because the lipid domains are to be modelled as patches on membrane surface, the domain boundaies will be curves on a
two-dimensional surface embedded in $\mathbb{R}^3$. Notice that if the interface is a geodesic, then it is a locally
straight line that does not curve to either domain it separates. How far an interface is from being a geodesic
is characterized by the geodesic curvature, which by definition is the curvature of the interface projected onto the 
tangent plane of the surface. The geodesic curvature will be an intrinsic property of {\it a binary lipid mixture}, similar
to the spontaneous mean curvature being an intrinsic property of {\it an individual lipid species}. We are then motivated to
define the curvature energy of the lipid domain boundary to be 
\begin{equation} \label{eqn:geocurv_eng}
G = \int_{C} k(H - H_0)^2 ds,
\end{equation}
where $C$ is the one-dimensional domain boundary contour on two-dimensional surface embedded in 
$\mathbb{R}^3$, $H$ is the geodesic curvature of the interface,
$H_0$ is the spontaneous geodesic curvature of the lipid mixture to be separated, and $k$ is the geodesic curvature
energy coefficient. This form of curvature energy is identical to the classical Canham-Evans-Helfrich energy in which
the integration in over the two-dimensional membrane surface and the mean curvature is adopted. The spontaneous geodesic 
curvature quantifies the geometry and molecular mismatch at the lipid interface. In the transitional region near the
interface two species of lipids coexist in a hybrid state, and a free energy for the hybrid packing of the lipids
(denoted by the subscripts 1 and 2 below) was proposed in recent theoretical studies \cite{BrewsterR2009a,BrewsterR2010a}:
\begin{equation} \label{eqn:bending_energy}
\mathcal{F} = k_1(L_1 - L_1^0)^2 + k_2(L_2 - L_2^0) ^2 + \gamma (L_1 - L_2)^2, 
\end{equation}
where $L_i$ is the length of the lipid chains in the transitional region and $L_i^0$ is the chain length in 
the equilibrium bulk. Coefficients $k_1,k_2$ are the free energetic costs of mismatch between two species and
their respective hybrids at the interface, and $\gamma$ is the energetic cost of mismatch hybrid chains of
different species. We choose $k_1 = k_2 = k$ for the purpose of simplification, but the generalization to
distinct $k$'s can be made easily. The domain curvature can be related to the lipid geometrical properties by the following
relation \cite{BrewsterR2010a}:
\begin{equation} \label{eqn:Vi}
V_i = L_i a_0 w_i \left(1 \pm \frac{w_i H}{2} \right), \quad i = 1, 2,
\end{equation}
where $V_i$ is the volume of the lipid chains, $w_i$ is the length that characterizes the molecular spacing of the lipid
head groups, and $a_0 = (w_1 + w_2)/2$ is the head group spacing in the hybrid region near the interface. Here the 
subtraction sign is chosen if the species is included in the microdomain, otherwise the addition sign shall be chosen.
The chain length in the equilibrium bulk state, $L_i^0$, can be computed by dividing the molecular volume using the 
head group area in the bulk state
\begin{equation} \label{eqn:Li0}
L_i^0 = \frac{V_i}{w_i^2}.
\end{equation}
Eqs.(\ref{eqn:Vi}-\ref{eqn:Li0}) allow us to represent the bending energy $\mathcal{F}$ in Eq.(\ref{eqn:bending_energy}) as
a quadratic function of the geodesic curvature $H$. The minimizer $H_0$ can be analytically calculated, and can be approximated
to the linear order of $V_d = V_1 - V_2$ and $w_d = w_1 - w_2$ by
\begin{equation} \label{eqn:H0}
H_0 = \frac{1}{w_T} \left [ \frac{(1 - 2 B)w_d }{(1+2 B) w_T} + \frac{2 B V_d }{(1+2 B) V_T}\right],
\end{equation}
where $V_T = (V_1 + V_2)/2, w_T = (w_1 + w_2)/2$ and $B = k/\gamma$. This approximate minimizer $H_0$ will 
appear in a more complicated form if $k_1 \ne k_2$, and calculations will show that $H_0$ is not sensitive to
the value of $B$ in its biologically relevant range \cite{BrewsterR2010a}. The minimizer $H_0$ is identified
as the {\it spontaneous geodesic curvature} of the corresponding binary lipid mixture.

\subsection{Phase field formulation of the geodesic curvature energy}

For a given binary lipid mixture in bilayer membrane, one can work to minimize the energy in Eq.(\ref{eqn:geocurv_eng}) 
to get the optimal boundary contour of surface lipid microdomains. This could be done, for example, by computing the shape
derivative following the general analytical approach outlined in \cite{DoganG2012a,MikuckiM2014a} and using this shape derivative 
as the boundary force to drive the motion of the initial boundary contour. Practical application of this approach, analytically
or numerically, might be hindered by the topological changes of the boundary contour, during which there are moments that
the surface is not smooth and its principal curvatures are not well defined. One can resort to phase field or level set 
methods with which the interface can be traced by evolving a function in a higher dimensional space. Here we adopt the
phase field approach, for which we define a phase field function $\phi$ associated with the boundary contour $C$ on the 
membrane surface $S$ as 
\begin{equation} 
\phi(x) = \tanh \left( \frac{d(x)}{\sqrt{2} \epsilon} \right), \label{eqn:phi_def}
\end{equation}
Here $d(x)$ is the signed geodesic distance from $x$ to the contour $C$, $d(x)<0$ in the interested domain enclosed by $C$,
$d(x)=0$ on $C$, and $d>0$ outside. Correspondingly, $\phi=-1$ and $1$ inside and outside, respectively, with the 
width of the transitional layer near $C$ determined by small $\epsilon>0$. This function $d(x)$ can be regarded as the 
generalization of the signed distance function used in Euclidean space $\mathbb{R}^n$ to the surface; some of the critical
critical properties of the latter can be shown true for $d(x)$. In Appendix \ref{sect:append} we provide details of relating
interface normal and geodesic curvature to this signed geodesic distance. With this definition of $\phi$ we can compute the geodesic 
curvature $H$ of its arbitrary level set, including $C$ 
where $\phi=0$. To facilitate the computation we first define 
$$ q(x) = \tanh \left(  \frac{x}{\sqrt{2} \epsilon} \right),$$
which gives 
$$ q'(x) = \frac{1}{\sqrt{2}} \left [ 1 - \tanh ^2 \left ( \frac{{ x}}{\sqrt{2}{ \epsilon}} \right ) \right ], ~
q''(x) = -\frac{1} {\epsilon} \tanh \left ( \frac{x}{\sqrt{2}\epsilon} \right ) \sech^2 \left ( \frac{x}{\sqrt{2}\epsilon} \right ).$$
Then 
\begin{eqnarray}
\nabla_{S} \phi & = & \frac{1}{\epsilon} q'(d(x)) \nabla_{S} d, \label{eqn:grad_phi}  \\
\Delta_{S} \phi & = & \nabla_S \cdot (\nabla_S \phi) = \frac{1}{\epsilon} q''(d(x)) |\nabla_{S} d |^2 + \frac{1}{\epsilon} q'(d(x)) \Delta_{S} d, \label{eqn:Lap_phi}
\end{eqnarray}
where $\nabla_S, \nabla_S \cdot $ and $\Delta_S$ are surface gradient, surface divergence, and Laplace-Beltrami operators on $S$, respectively.
The proof of the last equality (\ref{eqn:Lap_phi}) can be found in Appendix \ref{sect:append} as well, where we also proved that the tangent normal 
vector field for the level sets of $d(x)$ is $t = \nabla_S d$ and the geodesic curvature $H = \nabla_S \cdot t$. We then compute
$$ \nabla_S d = \frac{\epsilon}{q'(d(x))} \nabla_S \phi, \quad \Delta_S d = \frac{\epsilon}{q'(d(x))} \Delta_S 
\phi - \frac{q''(d(x))}{q'(d(x))} | \nabla_S d|^2,$$
It follows that the geodesic curvature is 
$$ H = \Delta_{S} d = \frac{\epsilon}{q'(d(x))} \Delta_{S} \phi - \frac{q''(d(x))}{q'(d(x))} \left |\frac{\epsilon}{q'(d(x))} \nabla_{S} \phi \right |^2.$$
Replacing $q(d(x))$ with $\phi$ and noticing that $\| \nabla_S d(x) \|=1$ (to be proved in Appendix \ref{sect:append}), we finally get the phase field 
representation of the geodesic curvature:
\begin{align}
H & = \frac{\sqrt{2} \epsilon}{1-\phi^2} \left( \Delta_{S} \phi + \frac{2\phi}{1-\phi^2} |\nabla_{S} \phi|^2 \right ) \nonumber \\
  & = \frac{\sqrt{2} \epsilon}{1-\phi^2} \left( \Delta_{S} \phi + \frac{1}{\epsilon^2} (1-\phi^2) \phi \right ). 
\end{align}
Correspondingly, we define a phase field representation of the geodesic curvature energy on the entire surface $S$, as follows:
\begin{equation}
G(\phi) = \int_{S} \frac{k \epsilon}{2} \left( \Delta_S \phi + \frac{1}{\epsilon^2} ( \phi + H_c \epsilon) ( 1- \phi^2)^2 \right)^2 ds, \label{eqn:G_phase_field}
\end{equation}
where $H_c = \sqrt{2} H_0$. For the phase field function $\phi$ defined in Eq.(\ref{eqn:phi_def}) one can show that this phase field curvature energy
formulation is equivalent to Eq.(\ref{eqn:geocurv_eng}). This nature is analogous to the equivalence between the Canham-Helfrich-Evans curvature 
energy and the phase field representation of the membrane elastic energy \cite{DuQ2005d}. The proof of the equivalence proceeds rather similarly,
and is omitted here. 

Our geodesic curvature energy model of surface microdomain formation can be regarded as a generalization of the mean-curvature driven membrane bending 
model from three-dimensional Euclidean space to two-dimensional surface. And the other way around, the mean-curvature driven membrane bending model 
can be viewed as as a special case of the curvature energy driven model on manifolds to Euclidean spaces. We expect our phase field based geodesic
curvature model can be applied to other physical systems that occur on manifolds where the formation and the topological change of surface patterns are 
driven by appropriate curvatures. These may include the change of spacetime topology curved by matter and energy in general 
relativity \cite{BaumgarteT1998a}.

\section{Gradient Flow of Geodesic Curvature Energy and its Numerical Treatments} \label{sect:numeric}

We follow the gradient flow of the curvature energy (\ref{eqn:G_phase_field}) to find its minimizer $\phi$ whose zero level set shall give the
the optimal lipid microdomain boundary. To derive the governing equation for the gradient flow we compute the first variation of $G$ with respect
to $\phi$:
\begin{equation}
\frac{\delta G}{\delta \phi} = k \left[ \Delta_{S} W - \frac{1}{\epsilon^2} (3 \phi^2 + 2H_c \epsilon \phi - 1) W \right], \label{eqn:var_G}
\end{equation} 
where 
$$ W = \epsilon \Delta_{S} \phi - \frac{1}{\epsilon} ( \phi + H_c \epsilon)(\phi^2 - 1).$$
The full expansion of this variation reads
\begin{eqnarray}
\frac{\delta G}{\delta \phi}  
    & = & 
	k \epsilon \Delta_{S}^2 \phi 
	+ \frac{k}{\epsilon}\left(2- 6 \phi^2 - 4 k H_c \epsilon \right) \Delta_{S} \phi 
	-   \left( \frac{6k}{\epsilon} \phi+ 2k H_c \right) | \nabla_{S} \phi|^2
	\nonumber \\ 
       & + & k \left( - \frac{2H_c^2 }{\epsilon}+ \frac{1}{\epsilon^3} \right) \phi 
	- \frac{ 3 k H_c}{\epsilon^2} \phi^2 
	- k\left( \frac{4}{\epsilon^3} - \frac{2  H_c^2}{\epsilon}\right) \phi^3 
	+ \frac{5kH_c}{\epsilon^2} \phi^4  
	+   \frac{3k}{\epsilon^3} \phi^5 
	\nonumber \\
	& +& \frac{k H_c}{\epsilon^2}.
\end{eqnarray}
During the energy minimization the amounts of each species of lipids shall be conserved; for this reason we consider the following constraint
\begin{equation}
A(\phi) = \int_{S} \phi(x)~ ds = \mbox{constant},
\end{equation}
whose first variation with respect to $\phi$ is
\begin{equation}
\frac{\delta A}{\delta \phi} = 1.
\end{equation}
We then introduce the following governing equation for the gradient flow of $\phi$   
\begin{equation} \label{eqn:gradient_flow_surface_pattern}
\frac{\partial \phi}{\partial t} =- \gamma \frac{\delta G}{\delta \phi} + \lambda \frac{\delta A}{\delta \phi},
\end{equation}
where $t$ is the pseudo-time, $\gamma$ is the diffusion coeffcient, and $\lambda$ is a Lagrangian multiplier used to ensure the conservation of $\phi$. 
The representation of $\lambda$ can be derived by integrating Eq. (\ref{eqn:gradient_flow_surface_pattern}) and noting 
that $\dd{\int_S \frac{\partial \phi}{\partial t} ~ ds = 0}$, thus
$$ 0 = -\int_S \frac{\delta G}{\delta \phi} ~ ds + \int_S \lambda ~ ds ,$$
and consequently,
$$ \lambda =  \frac{1}{|S|} \int_S \frac{\delta G}{\delta \phi} ~ ds,$$ 
which yields a forth-order nonlinear surface diffusion equation.
\begin{equation} \label{eqn:gradient_flow_surface_pattern_2}
\frac{\p \phi}{\p t} =- \gamma \frac{\delta G}{\delta \phi} + \frac{1}{|S|} \int_S \frac{\delta G}{\delta \phi} ~ds = -\gamma g + \lambda.
\end{equation}
Alternatively, one could derive a Cahn-Hilliard equation for the surface phase field function $\phi$ as
\begin{equation} \label{eqn:surface_CHE}
\frac{\p \phi}{\p t} = \gamma \Delta_S \left( \frac{\delta G}{\delta \phi} \right),
\end{equation}
which guarantees the conservation of $\phi$ and thus does not need a Lagrangian multiplier. However, it involves a sixth order 
surface derivative and thus is more complicated when the equation is to be solved numerically on a discretized surface $S$. 

To begin the time discretization of Eq.(\ref{eqn:gradient_flow_surface_pattern_2}) we first average the function $g$ over the 
current and the next steps $t_n$ and $t_{n+1}$, upon which the Crank-Nicolson scheme is applied:
\begin{equation} \label{eqn:scheme_1}
\frac{\phi_{n+1} - \phi_n }{\Delta t} + \gamma g (\phi_{n}, \phi_{n+1}) - \frac{1}{2} \big ( \lambda(\phi_{n+1}) + \lambda(\phi_n) \big) = 0,
\end{equation}
where the averaged function is defined by
\begin{align}
g(\phi_n, \phi_{n+1}) = & \frac{k}{2} \Delta_S (f(\phi_{n+1}) + f(\phi_n)) - \nonumber \\
                        & \frac{k}{2\epsilon^2} \left ( \phi_{n+1}^2 + \phi_{n} \phi_n + \phi_n^2 + H_c \epsilon(\phi_n + \phi_{n+1}) -1 \right )
  \left ( f(\phi_n) + f(\phi_{n+1}) \right ), 
\end{align}
with
\begin{equation}
f(\phi) = k \left (  \epsilon \Delta_S \phi - \frac{1}{\epsilon} (\phi + H_c \epsilon) (\phi^2 -1 ) \right ).
\end{equation}

To achieve convergence of this nonlinear implicit scheme at each time step, we define an inner iteration for computing $\phi_{n+1}$. The
solution during these inner iterations are denoted by $\psi_{m}$, which is expected to converges to $\phi_{n+1}$ as $m \rightarrow \8$.
We then replace $\phi_{n+1}$ in Eq.(\ref{eqn:scheme_1}) with $\psi_m$ and $\psi_{m+1}$ to get
\begin{equation} \label{eqn:scheme_2}
\frac{\psi_{m+1} - \phi_n }{\Delta t} + \gamma g (\phi_{n}, \psi_{m}, \psi_{m+1}) - \frac{1}{2} \big ( \lambda(\psi_{m}) + \lambda(\phi_n) \big) = 0,
\end{equation}
where the new average function is defined by
\begin{align}
g(\phi_n, \psi_n, \psi_{n+1}) = & \frac{k}{2} \Delta_S \tilde{f}(\phi_n, \psi_m, \psi_{m+1})  - \nonumber \\
                        & \frac{k}{2\epsilon^2} \left ( \psi_{m}^2 + \psi_{m} \phi_n + \phi_n^2 + H_c \epsilon(\psi_m + \phi_{n}) -1 \right )
  \left ( f(\psi_m) + f(\phi_{n}) \right ), 
\end{align}
with
\begin{equation}
\tilde{f}(\phi_n, \psi_m, \psi_{m+1}) = \frac{\epsilon}{2} \Delta_S (\psi_{m+1} + \phi_n) - \frac{1}{4 \epsilon} (\psi_m^2 + \phi_n^2 -2) (\psi_m + 
\phi_n + 2 H_c \epsilon).
\end{equation}

To efficiently solve the nonlinear Eq.(\ref{eqn:scheme_2}) we split the average functions into linear and nonlinear components, with the linear
component being the function of $\psi_{m+1}$ only and the nonlinear component being the function of $\phi_n$ and $\psi_m$, as follows:
\begin{align}
\tilde{f}_{lin}  = & \frac{\epsilon}{2} \Delta_S \psi_{m+1}, \\
\tilde{f}_{nlin} = & \frac{\epsilon}{2} \Delta_S \phi_{n} - \frac{1}{4 \epsilon} (\psi_m^2 + \phi_n^2 + 2H_c \epsilon) \\
g_{lin}  = & k \Delta_S \tilde{f}_{lin} (\psi_{m+1}), \\
g_{nlin}  = & k \Delta_S \tilde{f}_{nlin} (\psi_{m},\phi_n) - \nonumber \\
         & \frac{k}{2 \epsilon^2} \left( \psi_m^2 + \psi_m \phi_n + \phi_n^2 + H_c \epsilon (\psi_m + \phi_n) -1 \right)
(f(\psi_m) + f(\phi_n)).
\end{align}
The nonlinear implicit scheme is now given by
\begin{equation} \label{eqn:scheme_3}
\frac{\psi_{m+1} - \phi_n }{\Delta t} + \gamma ( g_{lin} + g_{nlin}) - \frac{1}{2} \big ( \lambda(\psi_{m}) + \lambda(\phi_n) \big) = 0,
\end{equation} 
which, after all terms of $\psi_{m+1}$ collected to the left-hand side, turns out to be
\begin{align} \label{eqn:scheme_4}
\left ( 1 + \frac{\gamma \Delta t k \epsilon}{2} \Delta_S^2 \right ) \psi_{m+1} = & \left [ \frac{1}{2} (\lambda(\psi_{m}) + 
\lambda(\phi_n)) - \gamma g_{nlin} \right ] \Delta t + \phi_n \nonumber \\
 = & \frac{1}{2} \big ( \lambda(\psi_{m}) + \lambda(\phi_n) \big) \Delta t - \frac{\Delta \gamma k \epsilon}{2} \Delta_S^2 \phi_n + \nonumber \\
 & \Delta t \gamma \Delta_S \left( \frac{k}{4 \epsilon} (\psi_m^2 + \phi_n^2 -2)(\psi_m + \phi_n + 2 H_c \epsilon) \right ) + \nonumber \\
 & \Delta t \gamma \Delta_S \left( \frac{k}{2\epsilon^2} (\psi_m^2 + \psi_m \phi_n + \phi_n^2 + H_c \epsilon(\psi_m + \phi_n) \right ) \cdot \nonumber \\
 & \left( f(\psi_m) + f(\phi_n) \right)  + \phi_n.
\end{align} 
This last equation is iterated over the inner index $m$ until the convergence is reached as signified by $\| \psi_{m+1} - \phi_{m} \| \le \beta$ 
for some desired tolerance. $\psi_m$ is initialized as $\phi_n$ for inner iterations at each time step. The convergent $\phi_m$ will be 
passed to $\phi_{n+1}$ to advance the computation to the next time step. 

We adopt a $C^0$ interior penalty discontinuous Galerkin (DG) method to solve the forth-order linear equation (\ref{eqn:scheme_4}) for $\psi_{m+1}$.
This method was initially developed for solve forth-order equations in two-dimensional Euclidean space \cite{BrennerS2012a}. When it is applied to 
solve equations on triangulated two-dimensional surface, an additional term arises on all edges, representing the mismatch of fluxes due to the
different outer normal directions of two neighbouring triangles at the shared edge. The order of convergence of the surface DG method with respect 
to the mesh refinement matches that observed on the original two-dimensional numerical experiments. A large time step is allowed in the computational
simulated presented below owning to the implicit time discretization. The inner iterations in general converge within five steps. The solutions of
a surface Cahn-Hilliard equation in \cite{DuQ2011a} is based on the reduction to two second-order equations, one for the order parameter and the other
for the chemical potential. The equations are then solved using standard surface finite element method along with an explicit time discretization, 
the latter suffers from a strict discrete energy stability condition and only very small time increments are allowed. An even larger
time step might be chosen by applying a recently developed exponential time differencing method \cite{WangX2016a}. With this method one needs to 
deliberately estimate the nonlinear component of $g(x)$ at every step so its slowly varying part can be combined with the linear component 
of $g(x)$, possibly further improving the discrete energy stability of the scheme.

\section{Computational Simulations} \label{sect:simulations}
In this section we shall apply the geodesic curvature energy variational model to simulate the lipid microdomain formation on 
three surfaces with different intrinsic geodesic curvatures. The dynamics of the microdomain formation in each set of simulations 
will be compared to those generated by the surface Cahn-Hilliard equation \cite{DuQ2011a}, the latter will produce large
surface domain as a result of minimizing the arc length of the domain boundaries.

We first consider a unit sphere on which a trigular mesh with $3963$ approximately uniformly distributed nodes. We choose 
$\varepsilon = 0.1, H_c = \frac{1}{0.3}, k = 0.01$, a time increment $\Delta t = 0.001$, and a random field of $\phi$. 
The results are compared side by side with those of the classical surface Cahn-Hilliard equation in 
Fig.~\ref{fig:1atm_3963}. Using a X-means clustering method \cite{PellegD_xmeans_2000} we are able to identify a number of microdomains 
whose radii are then calculated. The radius associated with each microdomain is approximately $0.23$. This means the curvature is 
approximately $\frac{1}{0.23}$, close to the specified spontaneous geodesic curvature. 
\begin{figure}[!htb]
\begin{center}
\includegraphics[angle=270,width=3.8cm]{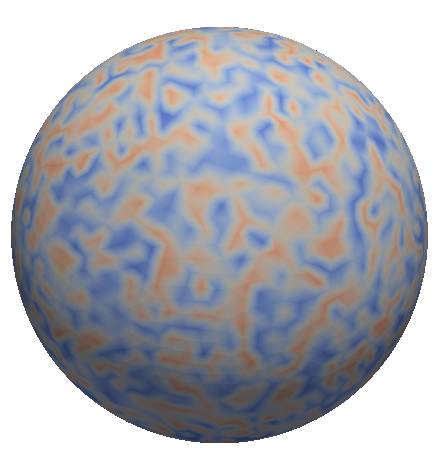}  \hspace{1.5cm} 
\includegraphics[angle=270,width=3.8cm]{1atm_3963_ini.png}  \\
\includegraphics[angle=270,width=3.8cm]{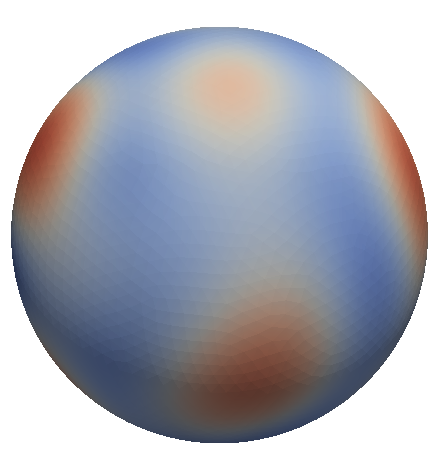} \hspace{1.5cm}
\includegraphics[angle=270,width=3.8cm]{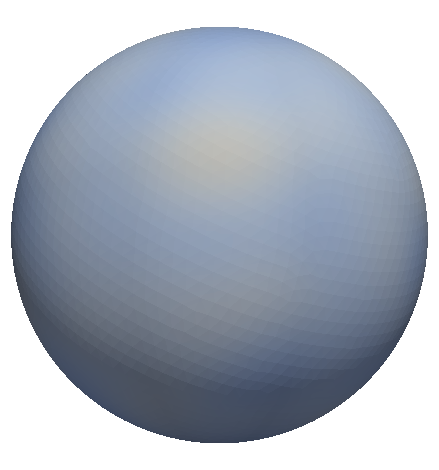}  \\
\includegraphics[angle=270,width=3.8cm]{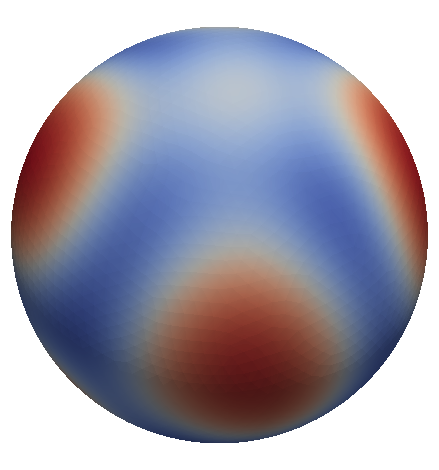} \hspace{1.5cm}
\includegraphics[angle=270,width=3.8cm]{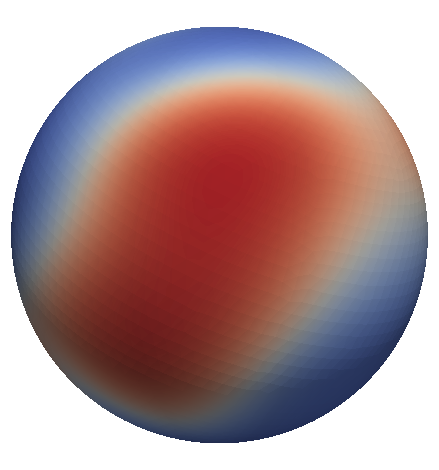}  \\
\includegraphics[angle=270,width=3.8cm]{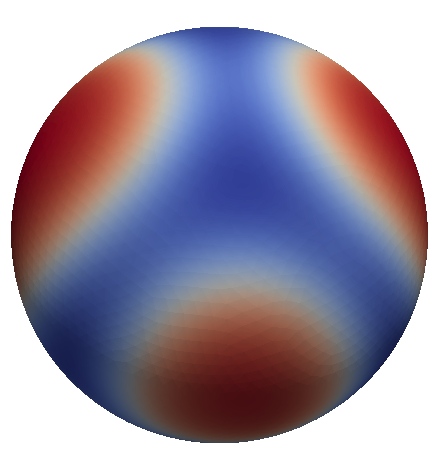} \hspace{1.5cm}
\includegraphics[angle=270,width=3.8cm]{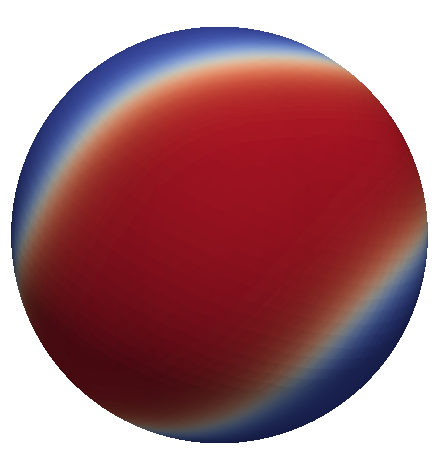}  
\end{center}
\caption{Simulations on the unit sphere with $3963$ nodes show the formation of local microdomains predicted by the geodesic curvature 
energy (left column) and domain separation predicted by the Allen-Cahn equation associated with the classical Ginzburg-Landau energy (right column) 
from the same initial random field (top row). Sampling time from top to bottom: $t=0, 1, 3$ and $8$.}
\label{fig:1atm_3963}
\end{figure}

The second set of simulations is conducted on a more complicated surface as shown in Fig. \ref{fig:3atm}. We choose this surface to be the 
molecular surface of three particles of unit radius respectively centered at $(0, 1, 0), (-0.864, -0.5, 0)$ and $(0.864, -0.5, 0)$. 
Such molecular surfaces are obtained by tracing the centers of spherical probe of the radius of the water molecule ($\approx 1.4$\AA, 
we are using dimensionless length is this work, though)
as the probe is rolling over the surfaces of the molecules \cite{MSMS}. This surface is quasi-uniformly meshed with $2974$ nodes and 
we set $\varepsilon = 0.1, H_c = \frac{1}{0.4}, k = 0.01$ and $\Delta t = 0.001$. Starting with a random initial field we finally identified six 
microdomains using the X-mean clustering method at the equilibrium state, whose radii are estimated. As seen in Fig.~\ref{fig:domain_radii}(left), 
the radii of the microdomains approximate the given spontaneous geodesic curvatures. 

\begin{figure}
\begin{center}
\includegraphics[angle=270,width=3.8cm]{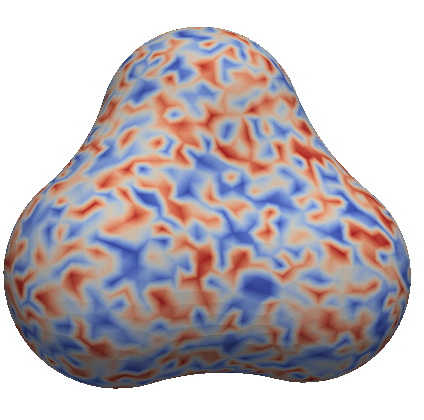}\hspace{1.5cm}
\includegraphics[angle=270,width=3.8cm]{3atm_ini.png} \\
\includegraphics[angle=270,width=3.8cm]{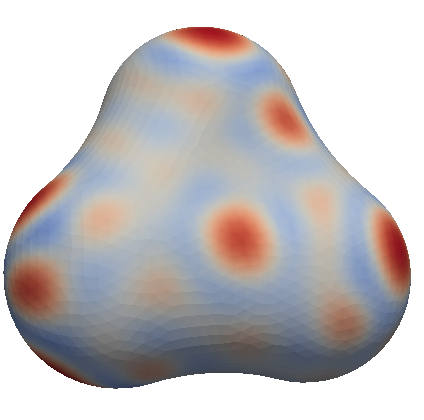}\hspace{1.5cm}
\includegraphics[angle=270,width=3.8cm]{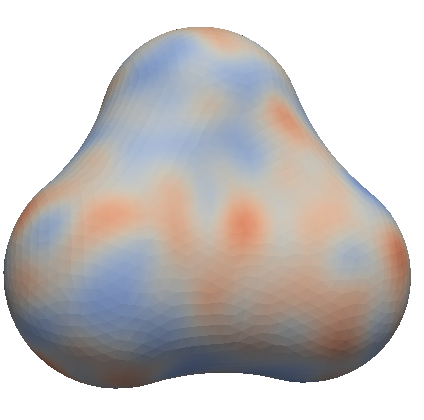} \\
\includegraphics[angle=270,width=3.8cm]{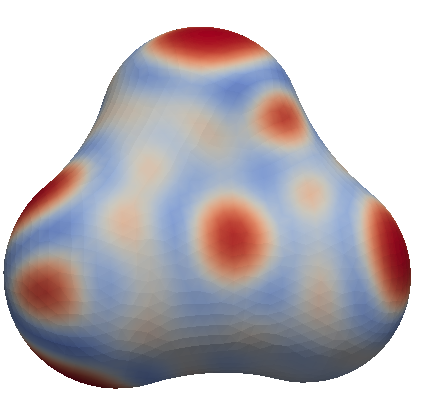}\hspace{1.5cm}
\includegraphics[angle=270,width=3.8cm]{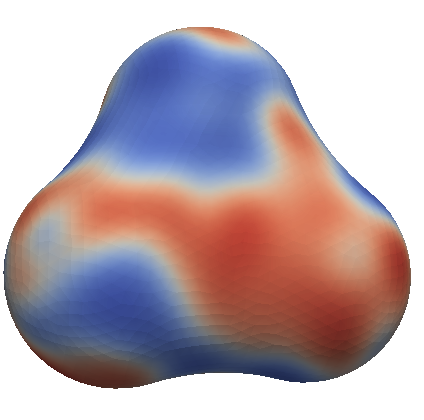} \\
\includegraphics[angle=270,width=3.8cm]{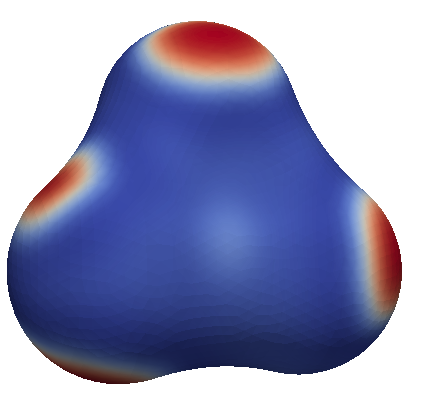}\hspace{1.5cm}
\includegraphics[angle=270,width=3.8cm]{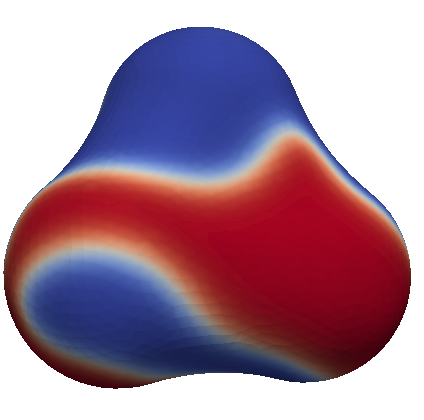} 
\end{center}
\caption{Simulations on the surface of a three-particle molecue show the formation of local microdomains predicted by the geodesic curvature 
energy (left column) and domain separation predicted by the Allen-Cahn equation associated with the classical Ginzburg-Landau energy (right column) 
from the same initial random field (top row). Sampling time from top to bottom: $t=0, 1, 3$ and $8$.}
\label{fig:3atm}
\end{figure}

\begin{figure}[!htb]
\begin{center}
\includegraphics[width=6.5cm]{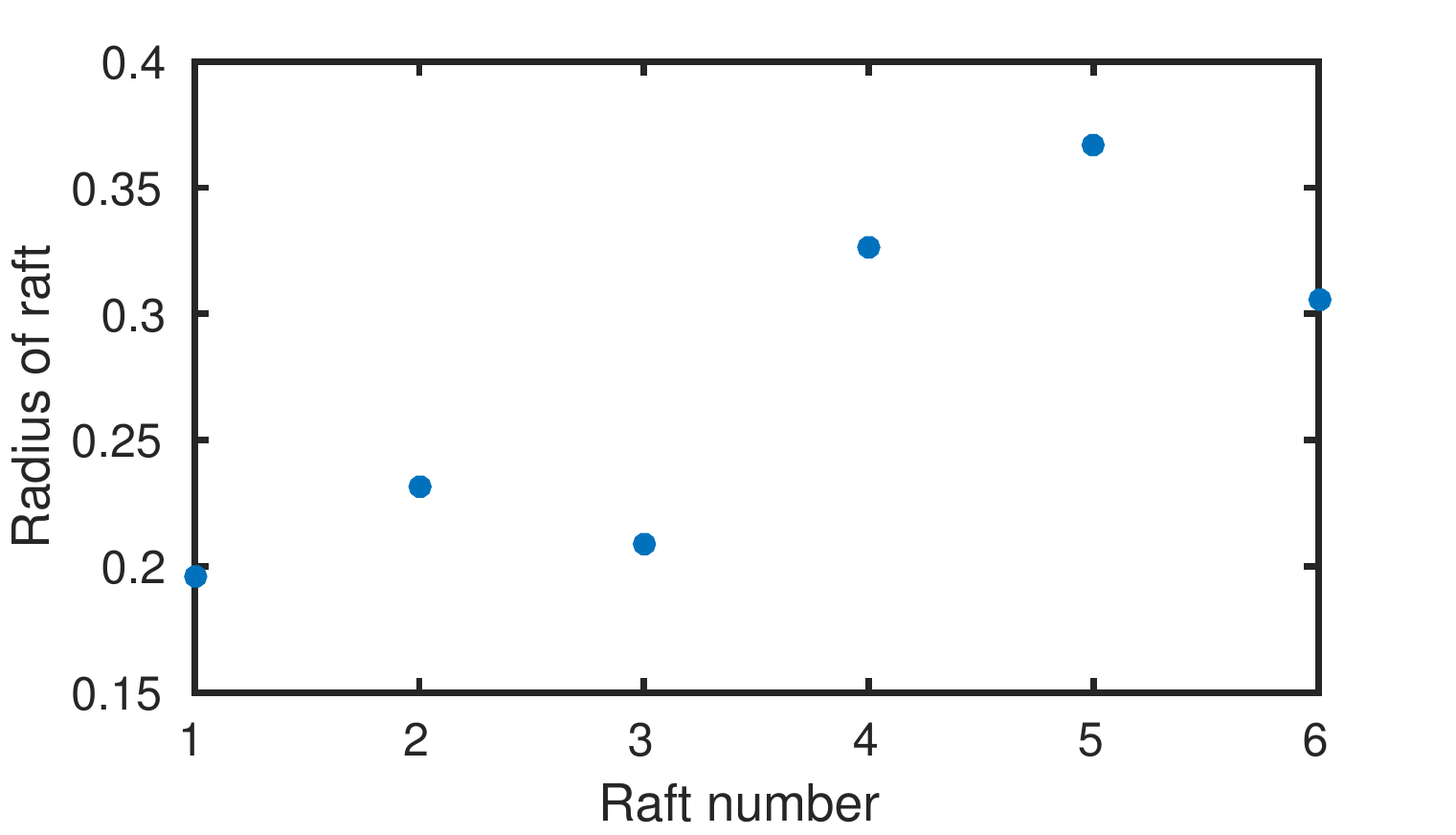}
\includegraphics[width=6.5cm]{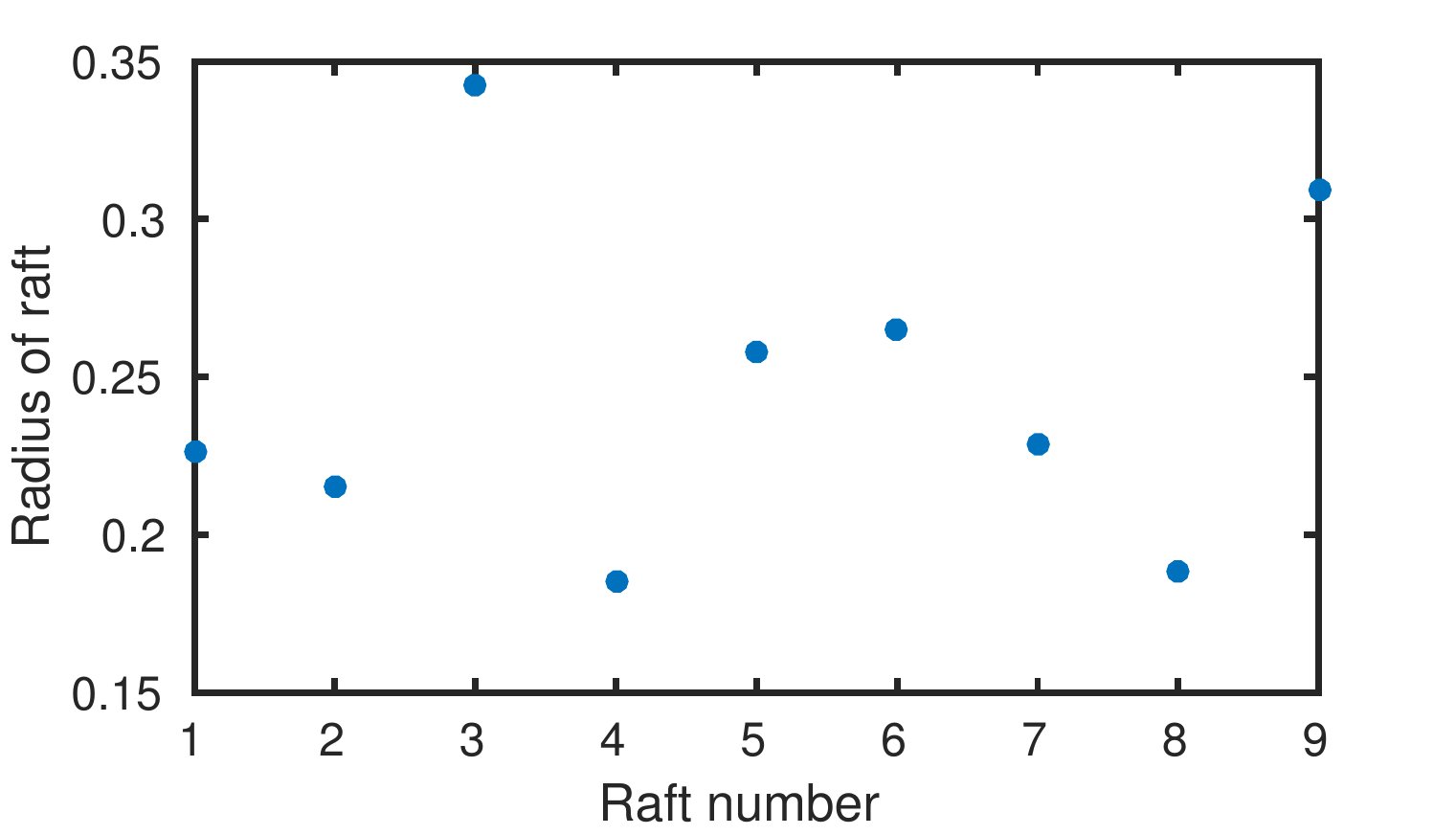}
\end{center}
\caption{The radii of the prominent 6 microdomains produced on surface of a three-atom molecule (left) and of the prominent 9 microdomains 
on the surface of a six-atom molecule.}
\label{fig:domain_radii}
\end{figure}

In the last set of simulations we choose the molecular surface of six particles of unit radius respectively centered 
at $(\pm 1,0,0),(0,\pm 1,0)$ and $(0,0,\pm)$. The quai-uniform surface mesh has $3903$ nodes and we 
choose $\epsilon = 0.1, H_c = \frac{1}{0.4}, k = 0.01$ and $\Delta t = 0.001$ for the simulation. The dynamics of the
microdomain formation is presented in Fig.~\ref{fig:6atm} along with that of the domain separation described by
the surface Cahn-Hilliard equation. One can see from Fig. \ref{fig:domain_radii}(right) that the largest raft radius obtained by the 
simulation is about $0.35$ which means the curvature of that raft is about $\frac{1}{0.35}$, a value close to given 
spontaneous geodesic curvature. 

\begin{figure}[!htb]
\begin{center}
\includegraphics[angle=270,width=3.8cm]{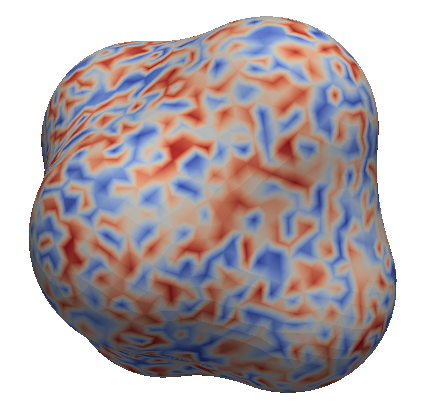} \hspace{1.5cm}
\includegraphics[angle=270,width=3.8cm]{6atm_ini.png}  \\
\includegraphics[angle=270,width=3.8cm]{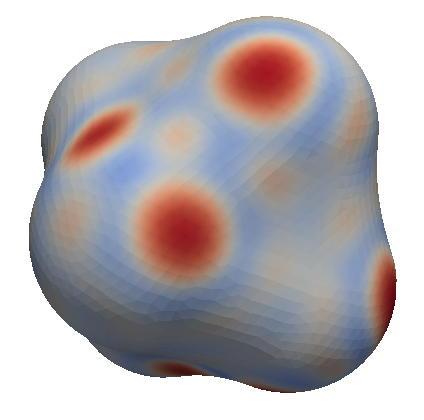} \hspace{1.5cm}
\includegraphics[angle=270,width=3.8cm]{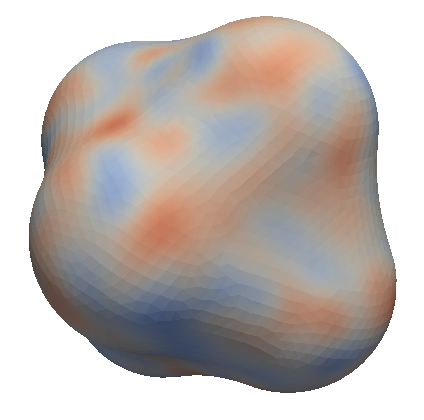}  \\
\includegraphics[angle=270,width=3.8cm]{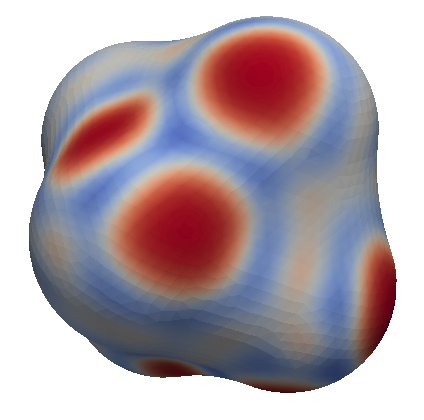} \hspace{1.5cm}
\includegraphics[angle=270,width=3.8cm]{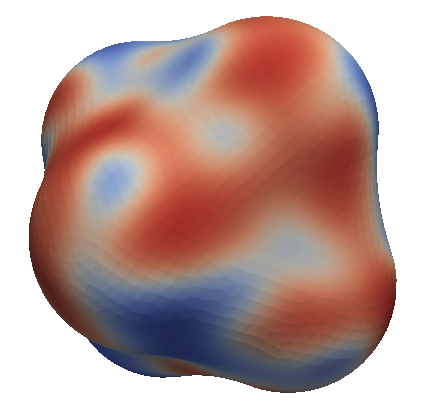}  \\
\includegraphics[angle=270,width=3.8cm]{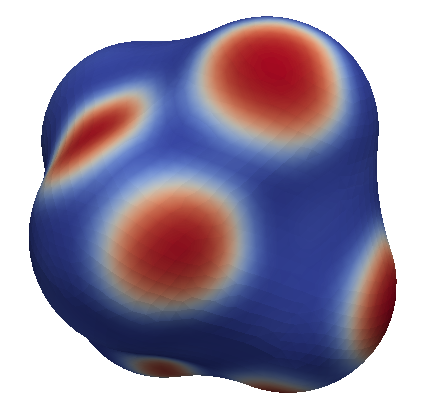} \hspace{1.5cm}
\includegraphics[angle=270,width=3.8cm]{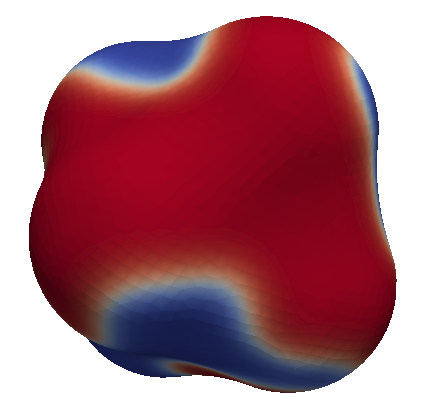}  
\end{center}
\caption{Simulations on the surface of a six-particle molecue show the formation of local microdomains predicted by the geodesic curvature 
energy (left column) and domain separation predicted by the Allen-Cahn equation associated with the classical Ginzburg-Landau energy 
(right column) from the same initial random field (top row). Sampling time from top to bottom: $t=0, 1, 3$ and $7$.}
\label{fig:6atm}
\end{figure}

The radii of the microdomains generated in our simulations are not exactly the given spontaneous geodesic curvature. Rather they 
are distributed around the given curvature. Apart from the numerical error in simulation and in X-mean clustering and radii 
estimate, this non-uniform distribution of domain radii is mostly related to the total quantity of the lipid phases 
in the initial random field. The initial quantity may not be exact to cover an integer number of microdomains with the 
given radius. However, the overall distribution of radii around the given radius of curvature demonstrates that our 
geodesic curvature model is capable of predicting the formation of microdomains that are caused by the geometrical and 
molecular mechanical mismatch of lipid mixtures. The predicted microdomains can be compared to the observed lipid rafts, 
and the boundaries of these microdomains can be identified to provide locations where specific proteins can aggregate. Coupling 
of our model of geodesic curvature driven microdomains formation to the localization of proteins will provide a very useful 
quantitative technique for studying the crucial roles of these proteins in high-fidelity signal transmission 
in cells \cite{LiH2013a,HockerH2013a}.

\section{Conclusion} \label{sect:summary}
In this work we investigated formation of microdomains on lipid membrane surfaces. We examine the 
geometrical and mechanical mismatch of lipids of different species on domain boundaries and introduce
the geodesic curvature of these boundaries in the energy characterization of the equilibrium state. 
This gives rise to the geodesic curvature energy model of the surface microdomain formation. A phase
field representation of the model is derived to ease the computational simulations of the model.
We develop a $C^0$ interior penalty surface finite element method along with an explicit/implicit time splitting scheme 
for solving the forth order nonlinear partial differential equation on surfaces. Numerical simulations on different 
surfaces demonstrate that our model and numerical techniques are able to produce the circular raft-like microdomains we are interested in. 
The generated microdomains have the radii corresponding to the given spontaneous geodesic curvature for the binary lipid mixture.
There are accumulated evidences that lipid rafts aid in vesicle budding and the deformation of the membrane. Currently many models are being 
produced to model vesicle budding and membrane deformation. This opens up a whole research area of connecting predicted microdomains using the
developed geodesic curvature model to other biophysical processes.

\section*{Acknowledgements}
This work has been partially supported by the Simons Foundation and National Institutes of Health through the grant 
R01GM117593 as part of the joint DMS/NIGMS initiative to support research at the interface of the biological and mathematical sciences.

\appendix

\section{Signed geodesic distance, normal vector, and geodesic curvature} \label{sect:append}

This section is devoted to the justification of a number of surface calculus identities used in the phase field formulation 
in this article. Although the counterparts of these identities in Euclidean spaces are well known, a rigorous derivation of these 
surface identities, to our knowledge, does not exist. We believe the justification of these identities is useful to scientists 
working on the various problems involving surface patterns and surface energies. Since the phase field function $\phi(x)$ is
an algebraic function of the signed geodesic function $d(x)$, it is sufficient to work with $d(x)$ in the derivation below.

Let $S$ be an oriented smooth surface without boundary in three-dimensional Euclidean space $\mathbb{R}^3$ and let $r(s)$ be a curve in $S$, 
parameterized by the arc length $s$. The Darboux frame of the curve $r(s)$ is defined by
$$ T(s) = r'(s), \quad u(s) = u(r(s)), \quad t(s) = u(s) \times T(s),$$
where $T(s), u(s), t(s)$ are respectively the unit tangent vector, the unit normal vector, and the unit tangent normal vector.
Let $d(x)$ be the signed geodesic function from the curve $r(s)$, following the convention that $d(x)<0$ in the 
surface region of interests enclosed by $r(s)$, zero on $r(s)$, and positive outside. One can define the similar Darboux frame on an arbitrary level set of $d(x)$ 
on $S$, giving rise to three vector fields $T(x), u(x)$ and $t(x)$ on $S$. 

We shall have 
\begin{theorem}
The tangent normal vector field $t(x)$ on $S$ is the normalized surface gradient of the signed geodesic distance function $d(x)$:
$$ t(x) = \nabla_S d(x).$$
\end{theorem}
\begin{proof}
Let $x_0$ be an arbitrary point on $S$ and let $r(s) = c$ be the level set of $d(x)$ passing through $x_0$ for a constant $c$, parameterized
by the arc length $s$. Associated with $r(s)$ at $x_0$ are the three vectors, $u(x_0), T(x_0), t(x_0)$ defined by the Darboux frame. Let $(\xi,\eta)$ be 
the two-dimensional orthogonal coordinate of the tangent plane of the surface $S$ originated at $x_0$, with $\xi$ aligned with $T(x_0)$. Furthermore, 
we consider a local Monge parameterization of $S$ around $x_0$, with the coordinate variables $(\chi_1, \chi_2)$. The linear map $M$ from $(\chi_1,\chi_2)$ to 
$(\xi,\eta)$ is such that
$$ \myvec{\xi;\eta}  = M \myvec{\chi_1;\chi_2}, \qquad \frac{d}{ds} \myvec{\xi;\eta}  = \nabla M \frac{d}{ds}\myvec{\chi_1;\chi_2}.$$ 
The signed geodesic distance function $d(x)$ can be given locally by either $d(\xi,\eta)$ or $d(\chi_1,\chi_2)$. 
The specific level set $r(s)=c$ can be locally given by $d(\chi_1(s),\chi_2(s))=c$ in the $(\chi_1,\chi_2)$ plane. 
It follows that 
$$ \frac{d d(\chi_1,\chi_2)}{ds} = \nabla_{(\chi_1,\chi_2)} d \cdot \frac{d}{ds} \myvec{\chi_1;\chi_2}= 0.$$
Using the mapping $M$ we will have at point $x_0$
\begin{align*}
\frac{d d(\xi,\eta)}{ds} & = \nabla_{(\xi,\eta)} d \cdot \frac{d}{ds} \myvec{\xi;\eta} \\
 & = (\nabla M)^T \nabla_{(\chi_1,\chi_2)} d \cdot (\nabla M)^{-1} \frac{d}{ds} \myvec{\chi_1;\chi_2} \\
 & = ((\nabla M)^T \nabla_{(\chi_1,\chi_2)} d )^T (\nabla M)^{-1} \frac{d}{ds} \myvec{\chi_1;\chi_2} \\
 & = \nabla_{(\chi_1,\chi_2)} d \cdot \nabla M (\nabla M)^{-1} \frac{d}{ds} \myvec{\chi_1;\chi_2} \\
 & = \nabla_{(\chi_1,\chi_2)} d \cdot \frac{d}{ds} \myvec{\chi_1;\chi_2} \\
 & = 0,
\end{align*}
suggesting that $\nabla_{(\xi,\eta)}d$ is orthogonal to the tangent vector $T(x_0)$ of $r(s) =c$ at $x_0$. Since $\nabla_{(\xi,\eta)}d$ is 
on the tangent plane it is orthogonal to surface with the normal vector $u(x_0)$, and hence it must be in the direction of the tangent normal vector
$t(x_0)$. 

Finally we check $\| \nabla_{(\xi, \eta)} d\|$. Since $r(s)=c$ is a level set and $\xi$ is aligned with the tangent direction of this level set it 
follows that $\partial d(\xi,\eta)/\partial \xi=0$. Now that $\eta$ is in the direction of the tangent norm, and given that $r(s)=c$ is a level set of 
the signed geodesic distance function, $\eta$ is indeed in the direction of the geodesic passing through $x_0$, thus $\partial d(\xi,\eta)/\partial \eta=1$.
This suggests that $\| \nabla_S d \| = \| \nabla_{\xi, \eta} d\|= 1$ and the theorem follows.
\end{proof}

After the establishment of the tangent normal vectors field $t(x)$ on $S$, we can compute the geodesic curvature of any level set of the 
signed geodesic distance function $d(x)$ by using the surface divergence of $t(x)$. Let $r(s)$ be a level set $d(x) = c$, parameterized by the arc 
length $s$.  One can associate a Frenet-Serret frame to $r(s)$, defined by 
$$ T(s) = r'(s), \quad N(s) = \frac{T'(s)}{\| T'(s) \|}, \quad B(s) = T(s) \times N(s),$$ 
where $T(s),N(s), B(s)$ are respectively the unit tangent vector, the Frenet normal vector, and the Frenet binormal vector. The Frenet normal
vector is also called the principal normal vector. The curvature $k$ of the curve $r(s)$ is 
$$ k = \| T'(s) \|.$$
The plane spanned by the orthogonal vectors $T(s),N(s)$ is called osculating plane. Let $\alpha$ be 
the angle between the osculating plane and the tangent plane. The geodesic curvature is defined to be the projection of the curvature to the 
tangent plane, i.e.,
$$ k_g = k \cos \alpha.$$
It can be shown that the geodesic curvature is indeed the ordinary curvature of the curve obtained by projecting $r(s)$ on to the 
tangent plane (P.261, \cite{Camo_DiffGeoCurvSurf_1976}). This last fact will be used to prove the following theorem:
\begin{theorem}
The geodesic curvature $k_g$ of the space curve $r(s)$ given by the level set $d(x)=c$ of the signed geodesic distance function equals 
to the surface divergence of the tangent normal vector field $t(x)$:
$$ k_g(x) = \nabla_S \cdot t(x).$$
\end{theorem}
\begin{proof}
Let $(\xi,\eta)$ be the orthogonal coordinate of a local Monge parameterization of the surface $S$ near the point $p$ such that
$\xi, \eta$ are respectively in the directions of the tangent and tangent normal vectors $T(p),t(p)$ at $p$, and $p$ itself is mapped
to the origin of the $(\xi,\eta)$ plane. Let the height of the surface be locally $\zeta = h(\xi,\eta)$. Under this parameterization the space 
curve $r(s)$ will be mapped to a curve on $(\xi,\eta)$ plane, denoted by $\eta = f(\xi)$. The space curve $r(s)$ can be parameterized by $\xi$ 
only by setting
$$ r(s)  = (\xi,f(\xi), h(\xi,f(\xi))).$$
Accordingly, the tangent vector $T(x)$ can be given as a function of $\xi$ by 
$$ T(\xi) = (1, f'(\xi), h_{\xi} + h_{\eta}f'(\xi))^T.$$
Since this tangent vector is identical to the element vector $(1,0,0)^T$ at $\xi=0$ by the choice of the Monge parameterization, it follows
that at $p$
$$ f'(0) = 0, \quad h_{\xi}(0,0) = 0.$$ 

We compute the geodesic curvature of the surface curve $r(s)$ at $p$ by evaluating the curvature of the plane curve $f(\xi)$ at $\xi=0$:
$$ k_g =\left. -\frac{f''(\xi)}{(1 + (f'(\xi))^2)^{3/2}} \right |_{\xi=0} = -f''(0),$$
where the negative sign is assigned so the orientation of the plane curve gives outer-pointing normal vector that is consistent with the
orientation of $(\xi,\eta)$ coordinate. It remains to show that $ \nabla_S \cdot t(p) =-f''(0)$.

We represent the tangent normal vector field $t(s) =  u(s) \times T(s)$ using the local Monge parameterization. For 
any point $y$ on the surface curve $r(s)$ we have 
$$ u(y) = \frac{(-h_{\xi}, -h_{\eta}, 1)^T}{\sqrt{1 + h^2_{\xi} + h^2_{\eta}}}$$
for some $\xi$ and then
\begin{align*}
t(y) & = u(y) \times T(y) \\
 & = \frac{(-h_{\xi}, -h_{\eta}, 1)^T}{\sqrt{1 + h^2_{\xi} + h^2_{\eta}}} \times 
\frac{(1,f'(\xi),h_{\xi} + h_{\eta} f'(\xi))^T}{\sqrt{1 + (f'(\xi))^2 + (h_{\xi} + h_{\eta} f'(\xi))^2}} \\
 & = \frac{1}{F(\xi)} \cdot (h_{\eta} \left( h_{\xi} + h_{\eta} f'(\xi)) - f'(\xi),1 + h^2_{\xi} + h_{\xi} h_{\eta} f'(\xi), - h_{\xi} f'(\xi) + h_{\eta} \right)^T,
\end{align*}
where 
$$ F(\xi) = \sqrt{1 + h^2_{\xi} + h^2_{\eta}} \cdot \sqrt{1 + (f'(\xi))^2 + (h_{\xi} + h_{\eta} f'(\xi))^2}.$$
To compute surface divergence we notice that
$$ \nabla_S \cdot t(x) = \frac{\partial t_1}{\partial \xi} + \frac{\partial t_2}{\partial \eta},$$
where $t_1,t_2$ are the components of the vector $t(x)$ respectively in $\xi,\eta$ directions. On one hand we shall have
$$  \frac{\partial t_2}{\partial \eta} = 0$$
at $p$ where $\xi =0$ for that $t_2=1$ achieves its maximum value as a component of a unit normal vector owning to $f'(0)=0,h_{\xi}(0,0)=0$ at $p$. 
On the other hand,
\begin{align*}
\frac{\partial t_1}{\partial \xi}  = & \frac{d t_1(\xi)}{d \xi} \\
 = & \frac{1}{F} \cdot \big [ - (h_{\xi \eta} + h_{\eta \eta} f'(\xi) ) h_{\xi} - ( h_{\xi\xi} + h_{\xi \eta} f'(\xi))h_{\eta} +  \\
  & \qquad  2(h_{\xi \eta} + h_{\eta \eta} f'(\xi)) h_{\eta} f'(\xi) + h^2_{\eta} f''(\xi) - f''(\xi) \big ] - \\
  & \frac{F_{\xi}}{F^2} \big [ 2(h_{\xi \xi} + h_{\xi \eta} f'(\xi)) h_{\xi} + 2 (h_{\xi \eta} + h_{\eta \eta} f'(\xi) ) h_{\eta} \big ] 
\big [(h_{\eta}h_{\xi} + h^2_{\eta} f'(\xi)) - f'(\xi) \big ]
\end{align*}
It follows again from $f'(0)=0,h_{\xi}(0,0)=0$ that $F(0) = 1$ and
$$  \frac{\partial t_1}{\partial \xi} =-f''(0),$$
as requested.
\end{proof}

The last fact to justify is the identity Eq.(\ref{eqn:Lap_phi}). Recalling that for a scalar $a$ and a vector field $v$ on the surface $S$ 
with outer normal vector $n$ it holds true that
\begin{align*}
\nabla_S \cdot (a v) & =  n \cdot \curl (n \times (a v)) \\
  & =  n \cdot \curl (a (n \times v)) \\
  & =  n \cdot \big [ a ~ \curl (n \times v) + \nabla a \times (n \times v) \big ] \\
  & =  a ~ n \cdot \curl (n \times v) + n \cdot \big [\nabla a \times (n \times v) \big ] \\
  & =  a \nabla_S v + n \cdot \big [(\nabla a \cdot v) n - (\nabla a \cdot n) v \big ],
\end{align*}
we shall have
\begin{align*}
\nabla_S \cdot \left( \frac{1}{\epsilon} q' \nabla_S d \right) & = \frac{1}{\epsilon} q' \nabla_S \cdot ( \nabla_S d ) + 
\nabla \left( \frac{1}{\epsilon} q' \right) \cdot \nabla_S d - \nabla \left( \frac{1}{\epsilon} q' \right) \cdot n (\nabla_S d \cdot n) \\
 & = \frac{1}{\epsilon} q' \nabla^2_S d + \frac{1}{\epsilon} \nabla q' \cdot \nabla_S d  \qquad \mbox{(for that $\nabla_S d \perp n$)}\\
 & = \frac{1}{\epsilon} q' \nabla^2_S d + \frac{1}{\epsilon} \big ( \nabla_S q' + n (\nabla q') \cdot n \big ) \cdot \nabla_S d \\ 
 & = \frac{1}{\epsilon} q' \nabla^2_S d + \frac{1}{\epsilon^2} q'' \nabla_S d \cdot \nabla_S d + \big( \frac{1}{\epsilon} \nabla q' \cdot n \big) \cdot (\nabla_S d \cdot n) \\
 & = \frac{1}{\epsilon} q' \Delta_S d + \frac{1}{\epsilon^2} q'' |\nabla_S d|^2
\end{align*}

\section*{References}
\bibliographystyle{plain}

\end{document}